\documentclass[12pt,onecolumn]{IEEEtran}

\usepackage{graphicx}
\usepackage{amsmath}
\usepackage{amssymb}
\usepackage{graphicx}
\usepackage{amsthm}
\usepackage{array}
\usepackage{latexsym}
\usepackage{subfigure}
\usepackage{multicol}
\usepackage{mathtools,cuted}
\usepackage{stfloats}
\usepackage{float}

\theoremstyle{definition}
\newcommand{\stab}{\ensuremath{\mathbb{S}}}
\newtheorem*{que*}{Question}

\newtheorem{cor}{Corollary}

\newtheorem{remark}{Remark}

\newtheorem{problem}{Problem}
\newcommand{\beq}{\begin{equation}}
\newcommand{\eeq}{\end{equation}}
\newcommand{\barr}{\left[\begin{array}}
\newcommand{\earr}{\end{array}\right]}

\newcommand{\bi}{\begin{itemize}}
\newcommand{\ei}{\end{itemize}}
\newcommand{\bnum}{\begin{enumerate}}
\newcommand{\enum}{\end{enumerate}}
\newcommand{\bc}{\begin{center}}

\newcommand{\stabS}{\mathbb{S}}
\newtheorem{thm}{Theorem}
\newtheorem{lem}{Lemma}

\begin{document}

\title{Stabilization Theory For Active Multi-port Networks }

\author{\IEEEauthorblockN{Mayuresh~Bakshi\IEEEauthorrefmark{2}~\IEEEmembership{Member,~IEEE,}, \thanks{\IEEEauthorrefmark{2}bmayuresh@ee.iitb.ac.in, Assistant Professor, Dept.\ of Engineering and Applied Sciences, VIIT Pune, India}
Virendra R Sule\IEEEauthorrefmark{1} \thanks{\IEEEauthorrefmark{1}vrs@ee.iitb.ac.in, Professor, Dept.\ of Electrical Engineering, IIT Bombay, India} \\ and Maryam Shoejai Baghini\IEEEauthorrefmark{3},~\IEEEmembership{Senior Member,~IEEE},\thanks{\IEEEauthorrefmark{3}mshoejai@ee.iitb.ac.in, Professor, Dept.\ of Electrical Engineering, IIT Bombay, India} }}

\maketitle

\begin{abstract}
 This paper proposes a theory for designing stable interconnection of linear active multi-port networks at the ports. Such interconnections can lead to unstable networks even if the original networks are stable with respect to bounded port excitations. Hence such a theory is necessary for realising interconnections of active multiport networks. Stabilization theory of linear feedback systems using stable coprime factorizations of transfer functions has been well known. This theory witnessed glorious developments in recent past culminating into the $H_{\infty}$ approach to design of feedback systems. However these important developments have seldom been utilized for network interconnections due to the difficulty of realizing feedback signal flow graph for multi-port networks with inputs and outputs as port sources and responses. This paper resolves this problem by developing the stabilization theory directly in terms of port connection description without formulation in terms of signal flow graph of the implicit feedback connection. The stable port interconnection results into an affine parametrized network function in which the free parameter is itself a stable network function and describes all stabilizing port compensations of a given network.    
\end{abstract}



\section{Introduction}
\IEEEPARstart {A}{ctive} electrical networks, which require energizing sources for their operation, are most widely used components in engineering. However operating points of such networks are inherently very sensitive to noise, temperature and source variations. Often there are considerable variations in parameters of active circuits from their original design values during manufacturing and cannot be used in applications without external compensation. Due to such uncertainties and time dependent variations, which cannot be modeled accurately, compensation of active networks or interconnections can lead to an unstable circuit or even greater sensitivity even if the component parts are stable. On the other hand interconnection of passive networks remains passive and stable. For this reason stability is never a consideration in passive network synthesis or design. Passive network theory and design thus enjoys a rich analytical framework devoid of the engineering complication of stability \cite{dary, andv}. On the other hand stability in analysis of active networks is an important property \cite{hayk} hence synthesis of active networks with stability is an important problem of circuit design. 

The purpose of this paper is to develop a systematic approach to interconnection of active linear time invariant (LTI) multi-port networks with the resolution of this stability issue in mind. Specifically, we address the following question. 
\begin{que*} \emph{
Given a LTI network $N$ with a port model $G$ at its nominal parameter values, what are all possible (LTI, active) networks $N_{c}$ with compatible ports, when connected to $N$ at ports with specified (series or parallel) topology, form a stable network?} 
\end{que*}
This question is analogous to that of feedback system theory, "what are all possible stabilizing feedback controllers of a LTI plant?" Such a question led to landmark new developments in feedback control theory in recent decades \cite{vids, doft, zhdg} broadly known under the $H_{\infty}$ approach to control. To motivate precise mathematical formulation of the problem we need to consider stability property of active networks and nature of interconnections at ports.

\subsection{Stabilization problem for network interconnection}
In active network theory two kinds of stability property have been well known \cite{chdk}, the \emph{short circuit stability} and the \emph{open circuit stability}. These can be readily extended to multi-port networks. Primarily these stability properties refer to the stability of the source to response LTI system associated at the ports where either an independent voltage or a current sources are attached and the network responses are the corresponding current or voltage reflected at the source ports respectively. Hence the problem of stabilization can be naturally defined as achieving stable responses at ports by making port interconnections with another active network. Such a stabilization problem has somehow never seems to have been addressed in the literature as far as known to authors. (Although such a problem appears to have made its beginning in \cite[chapter 11]{chdk}).

Modern algebraic feedback theory, whose comprehensive foundations can be referred from \cite{vids}, came up with the solution of the stabilization problem in the setting of feedback systems. However the formalism of feedback signal flow graph is not very convenient for describing port interconnections of networks. Although Bode \cite{bode} considered feedback signal flow graph to describe amplifier design the methodology did not easily carry further for multi stage and multi loop amplifier design. For single ports defining the loop in terms of port function is relatively simple as shown in \cite[chapter 11]{chdk}. Hence if stability of multi port networks is to provide a basis for stabilization problem of such networks, the stabilization problem must be formulated directly from port interconnection rules rather than signal flow graph rules. This forms the central motivation of the problem proposed and solved in this paper.

In recent times behavioral system theory \cite{powi}
considered problems of synthesis of control systems in which feedback control is subsumed in general interconnection between systems by defining linear relations between variables interconnected. Hence such interconnections are still equivalent to signal flow graph connections. The port interconnection in networks is however of special kind than just mathematical interconnection since connections between ports are defined between physical quantities such as voltages and currents and have to follow either series or parallel connection (Kirchhoff's) laws. Hence network connections at ports are instances of physically defined control rather than signal flow defined control which separates physical system from the logic of control. Control arising out of physical relationships between systems is also relevant in several other areas such as Quantum control systems, Biological control systems and Economic policy studies. Hence this theory of circuit interconnections should in principle be relevant to such other types of interconnected systems as well.   

\subsection{Background on systems theory, networks and coprime representation}
We shall follow notations and mathematical background of LTI circuits from \cite{chdk}. A driving point function of a LTI network is the ratio of Laplace transform of source and response physical signals which will always be currents or voltages at ports and which will be a rational function of a complex variable $s$ with real coefficients. In single port case a network function is always an impedance or admittance function. In the case of multi port network the driving point function is a matrix whose $(i,j)$-th entry is the driving point function with source at $j$-th port and response at $i$-th port. Hence the entries can also represent current and voltage ratios. 

We first need to make the precise well known assumption (as justified in \cite[chpater 10]{chdk}) that whenever a network is specified by its driving point function it has no unstable hidden modes. In the case of multi port networks this assumption can be made precise by assuming special fractional representation which depends on the notion of stable functions as follows. All the details of this approach are referred from \cite{vids}. The BIBO stable LTI systems have transfer functions without any poles in the RHP, the closed right half complex plane, (called stable proper transfer functions). This set of transfer functions form an algebra which is denoted as $\stabS$. General transfer functions $T$ of LTI systems are represented as $T=nd^{-1}$ where $n$, $d$ are themselves stable transfer functions and moreover $n$, $d$ are coprime. We refer \cite{vids,doft} for the theory of stable coprime fractional representation. Analogously we shall consider network functions always represented in terms of fractions of stable coprime functions. This is then an equivalent representation of network functions without hidden modes. For multiport network functions we shall consider the \emph{doubly coprime} fractional representation \cite{vids} as the hidden mode free representation. One more assumption we make for convenience is that whenever we consider network functions they will always be \emph{proper} functions without poles at infinity. Although this rules impedences and admittances of pure capacitors and inductors, in practice we can always consider these devices with leakage resistance and conductance hence their models are practically proper. Hence with this regularization to properness, any fractional representation $nd^{-1}$ will have the $d$ function in $\stabS$ without zeros at infinity. 

Finally, if $H=N_{r}D_{r}^{-1}=D_{l}^{-1}N_{l}$ are doubly coprime fractions of a network function $H$ then all functions $\tilde{H}$ in a neighbourhood of $H$ will be defined by doubly coprime fractions $\tilde{H}=\tilde{N}_{r}\tilde{D}_{r}^{-1}=\tilde{D}_{l}^{-1}\tilde{N}_{l}$ where $\tilde{N}_{r},\tilde{D}_{r},\tilde{D}_{l},\tilde{N}_{l}$ are in the neighbourhoods of $N_{r},D_{r},D_{l},N_{l}$ respectively. This way we shall draw upon the rich machinery of stable coprime fractional theory of \cite{vids} for formulating a stabilization theory for multi port interconnection.

The organization of this paper is as follows. Section 2 first reviews stability in single port networks and then the stabilization problem in case of single port networks is defined and solved. In section 3, the stability and stabilization problem in case of multiport networks is defined. The stabilization problem in case of interconnected two port network is then theoretically solved and the expression for the set of compensating networks stabilizing the interconnected network is derived. In section 4, a practical circuit example of two stage operational amplifier in unity feedback configuration is considered. This example is solved using coprime factorization approach to obtain set of  compensating networks in terms of a free parameter that stabilizes the given operational amplifier followed by conclusion in section 5.
\section {Stability and stabilization in single port networks}
We begin with the single port case. Two types of stability properties for a single port were defined in \cite{chdk} for LTI active networks. Consider a single port LTI network for which the port can be excited with an independent source of a known type. The stability property by definition depends on the type of this source.
\begin{enumerate}
\item \emph{\textbf{Short circuit stability.}} An independent voltage source $v_{s}$ is connected to an active network with driving point impedance $Z(s)$ at the port as shown in the Fig. \ref{fig001:shortcktstab}.
\begin{figure}[ht]
\centering
\includegraphics[width=6cm,height=2.8cm]{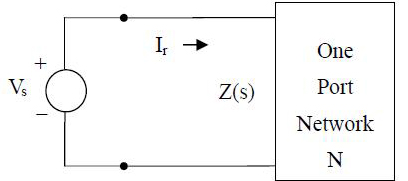}
\caption{One port network excited by a voltage source}\label{fig001:shortcktstab}
\end{figure} 
The network has all internal sources zero (or has zero stored energy). Let $V_{s}$ denotes the Laplace transform of this source voltage. The current $i_{r}$ at the port in transformed quantity denoted $I_{r} = Z(s)^{-1}V_{s}$ where $I_{r}$ is the Laplace transform of the response current $i_{r}$. The network is then said to be \emph{short circuit stable} if a bounded $v_{s}$ has bounded response $i_{r}$. Mathematically, this is equivalent to the condition, the network is short circuit stable iff $Z(s)^{-1}$ has no poles in RHP.
\item \emph{\textbf{Open circuit stability.}} An independent current source $i_{s}$ is connected to an active network with driving point admittance $Y(s)$ at the port as shown in the Fig.\ref{fig001:opencktstab}.
\begin{figure}[ht]
\centering
\includegraphics[width=6cm,height=2.8cm]{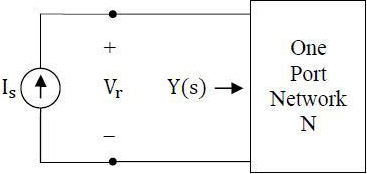}
\caption{One port network excited by a current source}\label{fig001:opencktstab}
\end{figure} 
The network has all internal sources zero (or has zero stored energy). Let $I_{s}$ denotes the Laplace transform of this source voltage. The voltage $v_{r}$ at the port in transformed quantity denoted $V_{r} = Y(s)^{-1}I_{s}$ where $V_{r}$ is the Laplace transform of the response voltage $v_{r}$ . The network is then said to be \emph{open circuit stable} if a bounded $i_{s}$ has bounded response $v_{r}$. Mathematically, this is equivalent to the condition, the network is short circuit stable iff $Y(s)^{-1}$ has no poles in RHP.
\end{enumerate}

\subsection {Stabilization problem in single port network}
Next we state stabilization problems in single port case. For a single port network, say $N$, connections at the port are series (or parallel) connections of impedance (or admittance). However, if the source at the port is a voltage (respectively current) then a parallel (respectively series) connection of a compensating network has no effect on the current (respectively voltage) in $N$ (respectively voltage across $N$). In such case, the connected network has no compensating or controlling effect on current (respectively voltage) in (or across) $N$. Hence the connection of a compensating network must be appropriate. This constraint leads to two different notions of stabilization. 

\subsubsection{Short circuit stabilization}
For stabilization of voltage fed impedance, say $Z$, it is required to change voltage across $Z$ which is possible only by a series connection of compensation impedance $Z_{c}$. Due to this compensation, the controlled current in $Z$ is $I_{r} = V_{s} / (Z +Z_{c})$. Following definition of short circuit stability, the stabilization problem envisages impedances $Z_{c}$ such that $(Z+Z_{c})^{-1}$ are stable. However a stronger requirement is chosen to define the short circuit stabilization as follows.

\begin{problem}[Short Circuit Stabilization]
\emph{Given a one port network with impedance function $Z$ fed by a voltage source, find all impedance functions $Z_{c}$ such that the impedance of the series connection $Z_{T} = (Z + Z_{c})$ satisfies
\begin{enumerate}
\renewcommand{\theenumi}{\roman{enumi}}
\item $Z_{T}^{-1}$ has no poles in RHP.
\item $\tilde{Z}_{T}^{-1}$ is a stable function where $\tilde{Z}_{T} = \tilde{Z} + Z_{c}$ for all $\tilde{Z}$ in a sufficiently small neighborhood of $Z$. 
\end{enumerate}
\ \ \ If above conditions are satisfied by $Z_{c}$ then it is called \emph{short circuit stabilizing compensator} of $Z$.
}
\end{problem}

First condition ensures stability of the interconnected network. The second condition is important in practice and requires that the compensator ensures stability of the interconnection over a sufficiently small neighbourhood of $Z$.

\subsubsection{Open circuit stabilization}
For compensation of current fed admittance, say $Y$, it is required to change current through $Y$ which is possible only by a parallel connection of compensating admittance $Y_{c}$. Due to this compensation, the controlled voltage across the port is $V_{r} = I_{s} / (Y +Y_{c})$. Following definition of open circuit stability, open circuit stabilization problem envisages finding all admittances $Y_{c}$ such that $(Y+Y_{c})^{-1}$ is stable. However a stronger requirement is chosen to define stabilization as follows. 

\begin{problem}[Open circuit stabilization]
\emph{Given a one port network with admittance function $Y$ fed by a current source, find all admittance functions $Y_{c}$ such that the admittance of the parallel connection $Y_{T} = (Y + Y_{c})$ satisfies
\begin{enumerate}
\renewcommand{\theenumi}{\roman{enumi}}
\item $ Y_{T}^{-1}$ has no poles in RHP.
\item $\tilde{Y}_{T}^{-1}$ is a stable function where $\tilde{Y}_{T} = \tilde{Y} + Y_{c}$ for all $\tilde{Y}$ in a sufficiently small neighborhood of $Y$. 
\end{enumerate}
\ \ \ If above conditions are satisfied by $Y_{c}$ then it is called \emph{open circuit stabilizing compensator} of $Y$.
}
\end{problem}

\subsection {Structure of the stabilizing compensator for single port open circuit stabilization}
We first describe the \emph{coprime fractional representation}. Consider the algebra $\stabS$ of stable proper network functions. A general admittance function of a single port network, say $Y$, is considered in the form $Y=nd^{-1}$ where $n,d$ belong to $\stabS$, $d$ has no zeros at infinity with the additional property that they are coprime i.e. have greatest common divisors which are invertible in $\stabS$. This is equivalent to the fact that there exist $x,y$ in $\stabS$ such that the following identity holds.
\begin{equation}
nx+dy=1
\end{equation}
Consider analogously a coprime fractional representation $Y_{c}=n_{c}d_{c}^{-1}$ for compensator network function $Y_{c}$ and $x_{c},y_{c}$ in $\stab S$ with the following identity.
\begin{equation}
n_{c}x_{c}+d_{c}y_{c}=1
\end{equation}
With this we have the following relationship between $Y$ and $Y_{c}$.

\begin{lem}\emph{
If a one port network with an admittance function, say $Y$, is fed by a current source is connected in parallel across an admittance $Y_{c}$ then the combined network is open circuit stable if and only if $nd_{c}+dn_{c}$ is a unit of $\stabS$.}
\end{lem}

\begin{proof}[ht]
The admittance of the parallel connection is as shown in the Fig. \ref{fig01:parconn}.

The parallel connection gives the combined admittance as $Y_{T}=Y+Y_{c}$ which has fractional representation as given below.
\begin{equation}
Y_{T}=\frac{nd_{c}+dn_{c}}{dd_{c}}
\end{equation}

For open circuit stability, $Y_{T}^{-1}$ as well as $\tilde{Y}_{T}^{-1}$ must be in $\stabS$ for all $\tilde{Y}$ in a sufficiently small neighbourhood of $Y$. Denote
\[
\tilde{\Delta}=\tilde{n}d_{c}+\tilde{d}n_{c}
\]
then 
\begin{equation}
\label{singleportparallel}
\tilde{Y}_{T}^{-1}=\frac{\tilde{d}d_{c}}{\tilde{\Delta}}
\end{equation}

It follows from equation (\ref{singleportparallel}) that if $\tilde{Y}_{T}^{-1}$ is stable then all roots of $\tilde{\Delta}$ in RHP are cancelled by RHP roots of $\tilde{d}d_{c}$. However over a neighbourhood of $n,d$ the pairs $\tilde{n},\tilde{d}$ are also coprime and hence do not have a common root in RHP. Since $d_{c}$ and $n_{c}$ are also copime, the only roots of $\tilde{\Delta}$ in RHP common with $d_{c}$ possible are those common between $d_{c}$ and $\tilde{d}$. But as $\tilde{d}$ varies over a neighbourhood of $d$ there can be no common roots with a fixed $d_{c}$ over the whole neighbourhood. Hence if $\tilde{Y}_{T}^{-1}$ is stable over a neighbourhood of $n,d$ then there is no possibility of RHP root cancellation between $\tilde{\Delta}$ and $\tilde{d}d_{c}$. Hence if $\tilde{Y}_{T}^{-1}$ is stable then $\tilde{\Delta}$ must not have a root in RHP or it must be a unit of $\stabS$. This proves necessity.

Conversely, if $\Delta=nd_{c}+dn_{c}$ is a unit then in a sufficiently small neighbourhood of $n,d$ all $\tilde{\Delta}$ are units in $\stabS$. Hence $\tilde{Y}_{T}^{-1}$ are stable functions in a neighbourhood. Hence sufficiency is proved.
\end{proof}
\begin{figure}
\centering
\includegraphics[width=4.8cm,height=2.8cm]{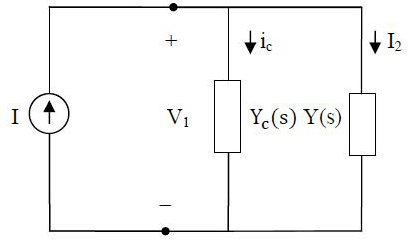}
\caption{One port network with compensating admittance connected in parallel}\label{fig01:parconn}
\end{figure}

It is worth noting here that for proving internal stability of feedback systems, the crucial formulation \cite{vids} of the stability of the map between two external inputs to two internal outputs is replaced by requiring stability over a neighbourhood of the given function $Y$. This also practically makes sense as given functions are never accurately same as real function, just as in control system a plant model is never an exact representation of the real plant. Although the two input-output formulation can be reconstructed for defining stability of the interconnection at the ports, we prefer the above approach of avoiding the feedback loop signal flow graph. 

In terms of given coprime representations of $Y$ as in above lemma we have a more special coprime fractional representations for $Y_{c}$,

\begin{cor}\label{cor1}\emph{
If a one port network with admittance function in fractional representation $Y=nd^{-1}$, is fed by a current source and an admittance $Y_{c}$ is connected in parallel across the port, then $Y_{c}$ stabilizes $Y$ 
iff there is a fractional representation $Y_{c}=n_{c}d_{c}^{-1}$ in $\stabS$ which satisfies $nd_{c}+dn_{c}=1$.
}
\end{cor}
\begin{proof}
If $nd_{c}+dn_{c}=1$ holds for some fractional representation $Y_{c}=n_{c}d_{c}^{-1}$ then we have $\Delta = 1$ a unit of $\stabS$. Hence $Y_{c}$ stabilizes $Y$ from above lemma.

Conversely, let $Y_{c}$ stabilizes $Y$ with coprime fractional representation $Y_{c}=n_{1c}d_{1c}^{-1}$, then $\Delta=nd_{1c}+dn_{1c}$ is a unit of $\stabS$ hence $Y_{c}=n_{c}d_{c}^{-1}$ where $n_{c}=n_{1c}\Delta^{-1}$, $d_{c}=d_{1c}\Delta^{-1}$ are also coprime and satisfy $nd_{c}+dn_{c}=1$. 
\end{proof}

The next theorem gives the set of all admittances $Y_{c}$ which form stable interconnection with a given admittance $Y$ when connected in parallel and fed by a current source.

\begin{thm}
If $Y=nd^{-1}$ is a coprime fractional representation of an admittance $Y$ with the identity $nx+dy=1 $
Then the set of all admittance functions $Y_{c}$ which form open circuit stabilizing parallel compensation with $Y$ are given by the fractional representation $Y_{c}=(y+qn)(x-qd)^{-1}$
where $q$ is an arbitrary element of $\stabS$ such that $x-qd$ has no zero at infinity.
\end{thm}

\begin{proof}
Suppose $Y_{c}=(y+qn)(x-qd)^{-1}$ for some $q$ in $\stabS$. Then,
\begin{equation}
\Delta=n(x-qd)+d(y+qn)=1
\end{equation}

This representation of $Y_{c}$ is a coprime fractional representation and by corollary (\ref{cor1}), it follows that $Y_{c}$ stabilizes $Y$.

Conversely, suppose $Y_{c}$ stabilizes $Y$ then from the above corollary we have a coprime fractional representation $Y_{c}=n_{c}d_{c}^{-1}$, where $n_{c}, d_{c}\in\stabS$ satisfy the following relation.
\begin{equation}
nd_{c}+dn_{c}=1
\end{equation}

Hence all solutions of $n_{c}$, $d_{c}$ in $\stabS$ of this identity with $d_{c}$ without a zero at infinity characterize coprime fractions of $Y_{c}$. Such solutions are well known (see \cite{vids, doft} for proofs) and are given by the formulas as below.
\begin{equation}
n_{c}=(y+qn),\mbox{  } d_{c}=(x-qd)
\end{equation}

This proves the formula claimed for all $Y_{c}$ which form an open circuit stable combination.
\end{proof}

Note that the admittance $Y$ and that of a stabilizing compensator $Y_{c}$ can not share a common pole or zero in RHP (called a non-minimum phase (NMP) pole or zero). Even their closer proximity in RHP would mean that the interconnected circuit has poor stability margin. Many such interpretations can be gathered from this algebraic characterization of stable interconnection of one port admittances which have been implicit part of knowledge of circuit designers or are new additions to this field.

\subsection{Structure of the stabilizing compensator for single port short circuit stabilization}
Analogous to the single port open circuit stabilization case, the formula for stabilizing $Z_{c}$ in case of short circuit stabilization problem can be stated and derived in similar manner. We shall thus state only the final theorem on the structure of the stabilizing impedance $Z_{c}$ for this case. 

In the present situation of short circuit stabilization, we have an impedance $Z$ fed by a voltage source $V_{s}$ the current is then $I_{r}=V_{s}/Z$. The current can be controlled only when we add another impedance $Z_{c}$ in series which changes the current to $V_{s}/(Z+Z_{c})$. Hence $Z_{c}$ is a short circuit stabilizing impedance iff $(Z+Z_{c})^{-1}$ is in $\stabS$. The structure of such a compensator is then given by the following.

\begin{thm}\emph{
If $Z=nd^{-1}$ is a coprime fractional representation of an impedance $Z$ with the identity $nx+dy=1$
Then the set of all impedance functions $Z_{c}$ which form short circuit stabilizing series compensation with $Z$ are given by the fractional representation. 
$Z_{c}=(y+qn)(x-qd)^{-1}$
where $q$ is an arbitrary element of $\stabS$ such that $x-qd$ has no zero at infinity.
}
\end{thm}

\begin{proof}
Proof readily follows from the open circuit stabilization case above by replacing $Y$, $Y_{c}$ by $Z$, $Z_{c}$ respectively along with their fractional representations and noting that in the present case stability of the interconnected network is equivalent to the fact that $(Z+Z_{c})^{-1}$ is in $\stab S$. 
\end{proof}

\section{Multi port stabilization for Bounded Source Bounded Response (BSBR) stability}
In case of multi port circuits, it is required to consider both, the open and short circuit stability, simultaneously due to existence of  independent voltage and current sources at the ports simultaneously. 
We first define the stability in the multi-port case. Consider a linear time invariant circuit represented by the following equation.
\begin{equation}
\label{multiport}
Y_r  = TU_s
\end{equation}
where $U_s$ denotes the vector of Laplace transforms of the independent sources at the ports and $Y_r$ denotes the vector of Laplace transforms of the responses at the ports (respecting indices). $T$ represents matrix of hybrid network functions between elements of $Y_r$ and $U_s$ respectively. We assume that $T$ is a proper rational matrix (with each of its elements having degree of numerator polynomial less than or at the most equal to the degree of the denominator polynomial) and has a formal inverse as a proper rational matrix.

We call such a circuit is \emph{Bounded Source Bounded Response} (BSBR) stable if for zero initial conditions of the network's capacitors and inductors, uniformly bounded sources have uniformly bounded responses. This is the case iff the hybrid network function (matrix) $T$ is stable i.e. $T$ has every entry belonging to $\stabS$. Let $M(\stabS)$ represent set of matrices of respective sizes whose elements belong to $\stabS$.

Consider a compensation network of same number and type of independent sources as the given network of (\ref{multiport}) to be compensated at all its port indices. In other words, we want to connect two ports of same index between the two networks only when both ports are either voltage fed or current fed. We can then connect the ports of the two networks at the index in either series or in parallel as shown in the Fig. \ref{portconnection}, as parallel (respectively series) connection of ports has no effect on the current (respectively voltage) in the individual circuits for the same voltage (respectively current) source. In other words a compensating network will have no effect on the response of a given network if connected in parallel (respectively series) at a voltage (respectively current) source. Hence we consider the interconnection of ports in series (respectively parallel) when the common source at the port is a voltage (respectively current) source. Let a compensating network has the hybrid function matrix $T_{c}$. (As in case of $T$, we assume $T_{c}$ to be proper rational with proper rational formal inverse.) Then for the source vector $U_{cs}$ the response vector $Y_{cr}$ in the compensating network is given by the following equation.
\begin{equation}
\label{eq4}
Y_{cr}= T_c U_{cs}
\end{equation}

\begin{figure}[ht]
\centering
\subfigure[Parallel connection for compensation when the source is current at ith port]{\includegraphics[width=3.8cm,height=4.8cm]{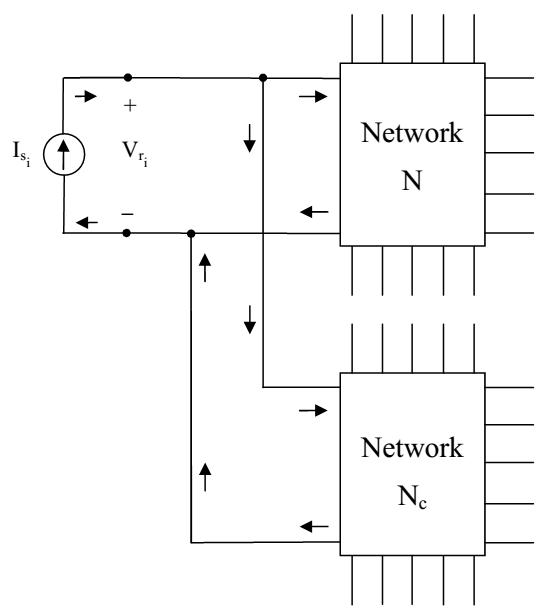} 
\label{fig:subfigure 1}}
\quad
\subfigure[Series connection for compensation when the source is voltage at ith port]{\includegraphics[width=3.8cm,height=4.8cm]
{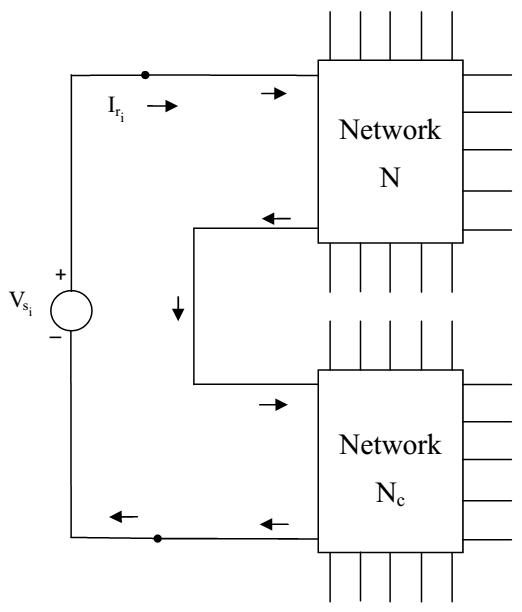}
\label{fig:subfigure 2}}
\caption{Diagram of two possible connections at a port}
\label{portconnection}
\end{figure}

Now if the two networks are connected as shown in the Fig. \ref{portconnection}, the independent source vector $\hat{U}_{s}$ applied to the interconnection distributes in the two networks as given by the following equation.
\begin{equation}
\hat{U}_{s}=U_{s}+U_{cs}
\end{equation}
while the common response vector $\hat{Y}_{r}$ of the two networks at the ports is given by the following equation.
\begin{equation}
\hat{Y}_{r}=TU_{s}=T_{c}U_{cs}
\end{equation}
Hence the source vectors reflected on ports of each network are given by the following equations.
\begin{equation}
U_{s}=T^{-1}\hat{Y}_{r}\ \ \ \ \ \mbox{   }U_{cs}=T_{c}^{-1}\hat{Y}_{r}
\end{equation}
Therefore, for the combined network we have the source response relationship given by the following equation.
\begin{equation}\label{TotalResponserelation}
\hat{Y}_{r}=(T^{-1}+T_{c}^{-1})^{-1}\hat{U}_{s}
\end{equation}
which is the hybrid representation of the interconnected network. It is thus clear that the interconnected network is BSBR stable iff the hybrid matrix of interconnection $(T^{-1}+T_{c}^{-1})^{-1}$ is in $M(\stabS)$. As in the single port case, we formally define the stabilization problem with additional restriction that the hybrid matrices of the interconnection arising from all $\tilde{T}$ in a neighborhood of $T$ are also stable.

\begin{problem}[Multi-port Hybrid Stabilization]
\emph{Given a multi-port hybrid matrix function $T$ of an LTI network, find all hybrid network function matrices $T_{c}$ of the compensating network connected as in figure \ref{portconnection} such that 
\begin{enumerate}
\item $\hat{T}=(T^{-1}+T_{c}^{-1})^{-1}$ is in $M(\stabS)$.
\item $(\tilde{\hat{T}}=\tilde{T}^{-1}+T_{c}^{-1})^{-1}$ is in $M(\stabS)$ for all $\tilde{T}$ in a neighbourhood of $T$.
\end{enumerate}
The matrix functions $T_{c}$ shall be called \emph{stabilizing hybrid compensators} of $T$.
}
\end{problem}

\subsection{Doubly coprime fractional representation}
For a comprehensive formulation of the multi-port stabilization we resort to the matrix case of coprime factorization theory over the $\stabS$ developed in \cite{vids}. This is the \emph{doubly coprime fractional representation} of proper rational functions over matrices $M(\stabS)$. For the proper rational network function $T$ the right coprime representation is $T=N_{r}D_{r}^{-1}$ where $N_{r}$, $D_{r}$ are matrices in $M(\stabS)$, $D_{r}$ is square, has no zeros at infinity and for which there exist $X_{l}$, $Y_{l}$ in $M(\stabS)$ satisfying the following identity.
\begin{equation}
X_{l}N_{r}+Y_{l}D_{r}=I
\end{equation}

Analogously, the left coprime representation is $T=D_{l}^{-1}N_{l}$ where $D_{l}$, $N_{l}$ are matrices in $M(\stabS)$, $D_{l}$ is square, has no zeros at infinity and for which there exist $X_{r}$, $Y_{r}$ in $M(\stabS)$ satisfying the following identity.
\begin{equation}
N_{l}X_{r}+D_{l}Y_{r}=I
\end{equation}
The doubly coprime representation of $T$ is then given as
\begin{enumerate}
\item $T$ is expressed by right and left fractions $T=N_{r}D_{r}^{-1}=D_{l}^{-1}N_{l}$ where $N_{r},D_{r},N_{l},D_{l}$ are matrices over $M(\stabS)$, $D_{r},D_{l}$ are square and have no zeros at infinity,
\item There exist matrices $X_{l},Y_{l}$ and $X_{r},Y_{r}$ in $M(\stabS)$ which satisfy the following equation.
\beq\label{dcf}
\left[\begin{array}{rr}
 X_{l} & Y_{l}\\
 D_{l} & -N_{l} 
\end{array}\right]
\left[\begin{array}{rr}
 N_{r} & Y_{r}\\
 D_{r} & -X_{r}
\end{array}\right]=
\left[\begin{array}{rr}
 I & 0\\
 0 & I
\end{array}\right]
\eeq
\end{enumerate}

We describe the doubly coprime fractional representation of a compensating network with hybrid network function $T_c$ by the respective matrices of fractions and identities by $N_{cr},D_{cr},N_{cl},D_{cl}$ and $X_{cr},Y_{cr},X_{cl},Y_{cl}$. It is also useful to recall that a square matrix $U$ in $M(\stabS)$ is called \emph{unimodular} if $U^{-1}$ also belongs to $M(\stabS)$. This is true iff $\det U$ is a unit or an invertible element of $\stabS$. 

Next, an open neighbourhood of $T$ is also specified in terms of the doubly coprime fractional representation of $T$. Any $\tilde{T}$ in a neighbourhood of $T$ is specified by a doubly coprime fractional representation 
with fractions $\tilde{T}=\tilde{N}_{r}\tilde{D}_{r}^{-1}=\tilde{D}_{l}^{-1}\tilde{N}_{l}$ and matrices $\tilde{X}_{l},\tilde{Y}_{l}$ and $\tilde{X}_{r},\tilde{Y}_{r}$ in $M(\stabS)$ satisfying the identities as given in equation (\ref{dcf}) in which the fractions $\tilde{N}_{r},\tilde{D}_{r}$, $\tilde{D}_{l},\tilde{N}_{l}$ are in respective neighbourhoods of the fractions of $T$.

In terms of the doubly coprime fractional (DCF) representation and the notion of neighbourhoods we have the preliminary.

\begin{thm}\emph{Consider the hybrid port interconnection as in the Fig. \ref{portconnection} of a given network $T$ with a compensating network $T_{c}$. Then the interconnection is BSBR stable (or $T_c$ stabilizes $T$) iff for a given doubly coprime fractions as above of $T$ there exist a doubly coprime fractions of $T_{c}$ that satisfy the following equation.
\begin{equation}\label{dcfrelationwithTc}
\left[\begin{array}{rr}
 D_{cl} & N_{cl}\\
 D_{l} & -N_{l} 
\end{array}\right]
\left[\begin{array}{rr}
 N_{r} & N_{cr}\\
 D_{r} & -D_{cr}
\end{array}\right]=
\left[\begin{array}{rr}
 I & 0\\
 0 & I
\end{array}\right]
\end{equation}
}
\end{thm}

\begin{proof}
The expression for the compensated network $\hat{T}$ in left (respectively right) coprime fractions of $T$ (respectively $T_{c}$) can be obtained as shown below.
\begin{align}
\label{thatinrcf}
\hat{T}&=(T^{-1}+T_{c}^{-1})^{-1}=[(N_{r}D_{r}^{-1})^{-1}+(D_{cl}^{-1}N_{cl})^{-1}]^{-1} \notag \\ 
&=N_{r}(\Delta_{r})^{-1}N_{cl}
\end{align}
where
\begin{equation}
\Delta_{r}=N_{cl}D_{r}+D_{cl}N_{r}
\end{equation}

Similarly, we can obtain the following equation by using left coprime fractions for $T$ ($D_{l}^{-1}N_{l}$) and right coprime fractions for $T_{c}$ ($N_{cr}D_{cr}^{-1}$). 
\beq\label{thatinlcf}
\hat{T}=N_{cr}(\Delta_{l})^{-1}N_{l}
\eeq
where
\begin{equation}
\Delta_{l}=N_{l}D_{cr}+D_{l}N_{cr}
\end{equation}

Analogous expression holds for the compensated network function $\tilde{\hat{T}}$ in terms of $\tilde{N}_{r}$, $\tilde{D}_{r}$ for all $\tilde{T}$ in a neighbourhood of $T$  as given below.
\begin{equation}
\tilde{\hat{T}}=\tilde{N}_{r}(\tilde{\Delta}_{r})^{-1}N_{cl}
\end{equation}
where 
\begin{equation}
\tilde{\Delta}_{r}=N_{cl}\tilde{D}_{r}+D_{cl}\tilde{N}_{r}
\end{equation}

If $T_{c}$ stabilizes $T$ then $\tilde{\hat{T}}$ is in $M(\stabS)$ for all $\tilde{T}$ in a neighbourhood of $T$. If $\tilde{\Delta}_{r}$ has any RHP zeros when $\tilde{T}$ varies in a neighbourhood of $T$ then the poles of $\tilde{\hat{T}}$ in RHP can appear only from such zeros. Since such zeros also vary continuously with parameters of $\tilde{T}$ in an open neighbourhood and $N_{cl}$ has constant parameters, it follows that $\tilde{\hat{T}}$ is in $M(\stabS)$ iff the matrix
$
\tilde{N}_{r}(\tilde{\Delta}_{r})^{-1}
$
belongs to $M(\stabS)$ for all $\tilde{T}$ in a neighbourhood of $T$. Let $\tilde{z}$ be a zero in RHP of $\tilde{\Delta}_{r}$ then 
there is a vector $\tilde{v}$ over $\stabS$ (also varying with parameters) such that
\begin{equation}\label{rhpzeronr}
\tilde{N}_{r}(\tilde{z})\tilde{v}(\tilde{z})=(\tilde{\Delta}_{r})(\tilde{z})\tilde{v}(\tilde{z})=0
\end{equation}
which is equivalent to both $\tilde{N}_{r}$ and $\tilde{\Delta}_{r}$ having a common RHP zero at $\tilde{z}$ for all $\tilde{T}$ in a neighbourhood of $T$. But then this implies the following relation for all $\tilde{T}$.
\begin{equation}
\label{rhpzeronr1}
(\tilde{\Delta}_{r})(\tilde{z})\tilde{v}(\tilde{z})=N_{cl}(\tilde{z})\tilde{D}_{r}(\tilde{z})\tilde{v}(\tilde{z})=0
\end{equation}
Since $N_{cl}$ has constant parameters, its zeros are stationary for variations of $\tilde{T}$. Hence the equation (\ref{rhpzeronr1}) simplifies to the following equation.
\begin{equation}
\label{rhpzeronr2}
\tilde{D}_{r}(\tilde{z})\tilde{v}(\tilde{z})=0
\end{equation}
for all $\tilde{T}$ in a neighbourhood of $T$. However equation (\ref{rhpzeronr2}) along with equation (\ref{rhpzeronr}) mean that $\tilde{N}_{r},\tilde{D}_{r}$ are not right coprime in any neighbourhood of $T$. Since the coprime fractions remain coprime in an open neighbourhood of $T$, this is a contradiction. This proves that $\tilde{\Delta}_{r}$ has no zeros in RHP or that $\tilde{\Delta}_{r}$ is unimodular in a sufficiently small neighbourhood of $T$. Using identical arguments it follows that $\tilde{\Delta}_{l}$ is also unimodular. In partucluar it follows that $\Delta_{l}$ and $\Delta_{r}$ are unimodular. Now if a stabilizing $T_{c}$ is represented by right and left coprime fractions $N_{cr},D_{cr}$ and $D_{cl},N_{cl}$ then the new fractions, right fractions $N_{cr}\Delta_{l}^{-1},D_{cr}\Delta_{l}^{-1}$ and left fractions $\Delta_{r}^{-1}D_{cl},\Delta_{r}^{-1}N_{cl}$ satisfy the relations (\ref{dcfrelationwithTc}). This proves the necessity part of the claim.

Now, let the relations (\ref{dcfrelationwithTc}) be satisfied between the DCFs of $T$ and $T_c$. Then the interconnection network function is
\beq\label{thatexpressionindcf}
\hat{T}=N_{cl}N_{r}=N_{l}N_{cr}
\eeq
since $\Delta_{l}=\Delta_{r}=I$. Hence $\hat{T}$ is BSBR stable. On the other hand for $\tilde{T}$ a sufficiently small open neighbourhood of $T$ the perturbed fractions $\tilde{N}_{r},\tilde{D}_{r}$, $\tilde{D}_{l},\tilde{N}_{l}$ perturb $\tilde{\Delta}_{r}$ and $\tilde{\Delta}_{l}$ from identity but they still remain unimodular. Hence the interconnection function
\[
\hat{\tilde{T}}=\tilde{N}_{r}(\tilde{\Delta}_{r})^{-1}N_{cl}=\tilde{N}_{cl}(\tilde{\Delta}_{l})^{-1}N_{r}
\]
has no poles in RHP hence is BSBR stable for all perturbations in a sufficiently small neighbourhood. This shows that $T_{c}$ stabilizes $T$. This proves sufficiency.
\end{proof}

\begin{remark}\emph{Entire proof above can also be written starting from the left coprime fractions for $T$ and the expression (\ref{thatinlcf}) for the interconnection function. At the same time above theorem can also be expressed starting with DCF of $T^{-1}$ and establishing the structure of $T_{c}^{-1}$ which are just another hybrid port matrix functions of these networks.
}
\end{remark}

The structure of stabilizing compensators $T_c$ now follows from the equation (\ref{dcfrelationwithTc}) in terms of the DCF of $T$ as follows.

\begin{cor}\emph{
Given a DCF (\ref{dcf}) of $T$ the set of all stabilizing compensators $T_c$ are given by any of the following alternative formulae.
\beq\label{Tcparametrization}
\begin{array}{rcl}
T_{c} & = & (X_{l}-QD_{l})^{-1}(Y_{l}+QN_{l})\\
T_{c} & = & (Y_{r}+N_{r}Q)(X_{r}-D_{r}Q)^{-1}
\end{array}
\eeq
for all $Q$ in $M(\stabS)$ such that functions $\det (X_{l}-QD_{l})$ and $\det (X_{r}-D_{r}Q)$ have no zero at infinity.
}
\end{cor}

\begin{proof}
The stated formulas are all solutions of the identity (\ref{dcfrelationwithTc}) which shows the relationship between $T$ and a stabilizing $T_c$. The conditions on zeroes of denominator fraction matrices is to ensure that these matrices are proper when inverted.
\end{proof}
\section{Multiport network stabilization Example}
We now show an example of stabilization of a practical circuit of two stage operational amplifier in unity feedback configuration. The equivalent circuit for this two stage op amp without a compensating network is shown in the Fig.\ref{equivalentckt}(a). It is required to find a compensating network $T_{c}$ such that the interconnection is stable.


The compensating network can be connected across the two amplifier stages and we can consider a port with a pair of terminals formed due to its connection. Thus the equivalent circuit can be redrawn as shown in the Fig.\ref{equivalentckt}(b). 
\begin{figure}[!ht]
\centering
\subfigure[Small signal equivalent circuit of two stage opamp]{\includegraphics[width=8.4cm,height=3.2cm]{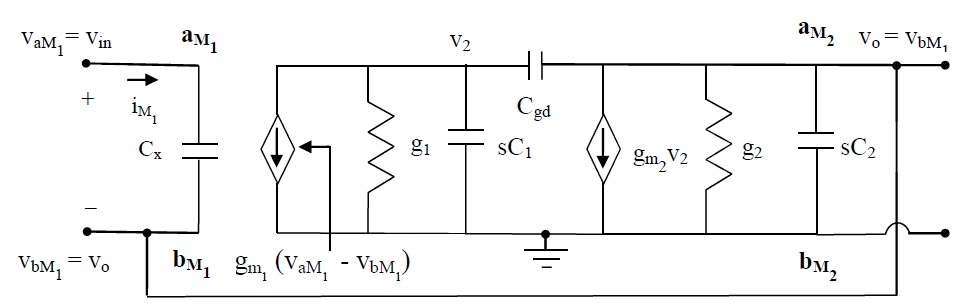}}
\label{eqckt1}
\quad
\subfigure[Small signal equivalent circuit of two stage opamp as a two port network]
{\includegraphics[width=8.2cm,height=4cm]{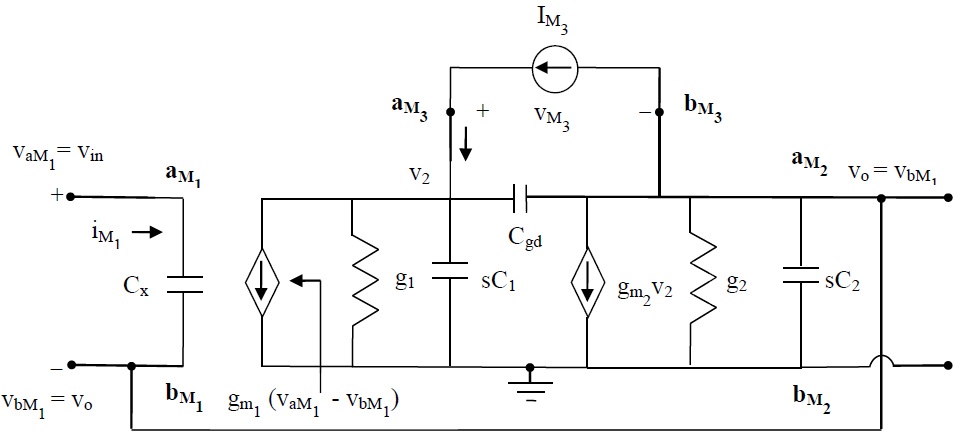}}
\label{fig12:subfigure2}
\caption{Two stage opamp with small signal equivalent circuit}
\label{equivalentckt}
\end{figure}


Let $Y_r $=
$ \begin{bmatrix}
I_{\scriptscriptstyle M_1} (s) \\
V_{\scriptscriptstyle M_3} (s)
\end{bmatrix}$
be vector of Laplace transforms of the responses and
$ U_s=\begin{bmatrix}
V_{\scriptscriptstyle a M_1}(s) \\
I_{\scriptscriptstyle M_3} (s)
\end{bmatrix}$
be vector of Laplace transforms of the independent sources.Thus we have the following matrix equation between the excitation and response signals.
\begin{equation}
\label{doubly1}
\begin{bmatrix}
I_{\scriptscriptstyle M_1} (s) \\
V_{\scriptscriptstyle M_3} (s)
\end{bmatrix}
=
\begin{bmatrix}
T_{11} & T_{12} \\
T_{21} & T_{22}
\end{bmatrix}
\begin{bmatrix}
V_{\scriptscriptstyle a M_1}(s) \\
I_{\scriptscriptstyle M_3} (s)
\end{bmatrix}
\end{equation}

This is equivalent to $Y_r=TU_s$ where elements of matrix $T$ can be computed using the parameter values associated with the equivalent circuit. The circuit can be simplified and solved using Kirchhoff's laws so that the elements of matrix $T$ are given as,
\begin{align}
T_{11}&=sC_{x}\Big[1+ \frac {g_{\scriptscriptstyle m_1}(sC_{gd}-g_{\scriptscriptstyle m_2})}{D_{1}}\Big]  \\
T_{12}&=\Big[\frac{-sC_{x}}{sC_2+g_{\scriptscriptstyle m_2}+g_2}\Big] \Big[-1+\frac{N_{1}(sC_{gd}-g_{\scriptscriptstyle m_2})}{D_{1}}\Big]  \\
T_{21}&= \frac {-g_{m_1}[sC_2+g_{\scriptscriptstyle m_2}+g_{\scriptscriptstyle_2}]} {D_{1}}   \\
T_{22}&= \frac {s(C_1+C_2)+(g_{\scriptscriptstyle m_2}-g_{\scriptscriptstyle m_1}+g_2+g_1)}{D_{1}}  \end{align} 
where
\begin{align*}
D_{1}&=  [(C_1+C_2)C_{gd}+C_1C_2]s^2 +[C_2g_1+C_1g_2+(g_{\scriptscriptstyle m_2}-g_{\scriptscriptstyle m_1}+g_1+g_2)C_{gd}]s 
+[g_{\scriptscriptstyle m_1}g_{\scriptscriptstyle m_2}+g_1g_2] \\
N_{1}&=s(C_1+C_2)+(g_{\scriptscriptstyle m_2}-g_{\scriptscriptstyle m_1}+g_2+g_1)
\end{align*}

But this gives some of the elements of the transfer function matrix $T$ (such as $T_{11}$) as improper with degree of numerator polynomial greater than the degree of denominator polynomial. This poses difficulty in matrix inversion. 

This computational difficulty can be resolved by adopting the regularization procedure which include adding the resistors either in series (or in parallel) of appropriate values at the ports so that none of the elements of the transfer function matrix $T$ are improper. The modified equivalent circuit after regularization is as shown in the Fig.\ref{fig06:name}.
\begin{figure}[H]
\centering
\includegraphics[scale=0.32]{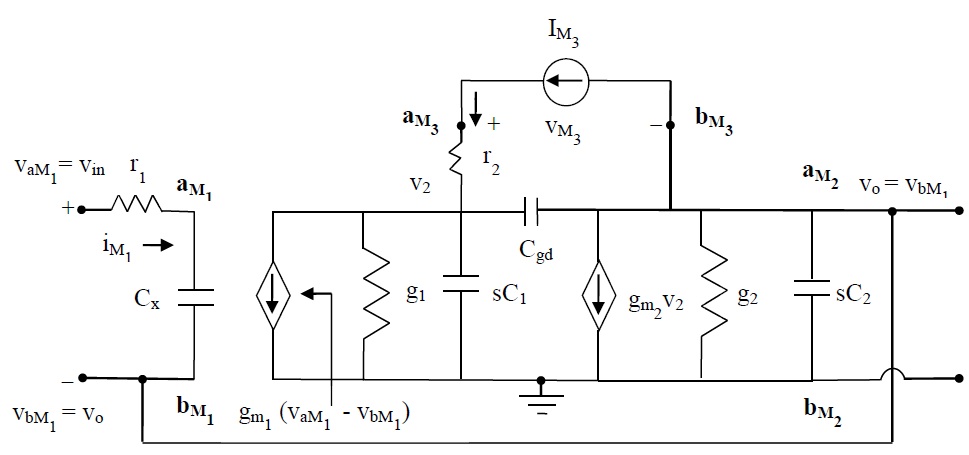}
\caption{Modified Equivalent circuit of two stage opamp }\label{fig06:name}   
\end{figure}

Simplifying and solving the modified equivalent circuit, the elements of matrix $T$ are as given below.
\begin{align}
T_{11}&=\Big(\frac{sC_{x}}{sC_{x}r_{1}+1}\Big)\Big[1+ \frac {g_{ m_1}(sC_{gd}-g_{m_2})}{D_{1}}\Big] \\
T_{12}&=N_{1}.\Big [-1+(g_{m_2}-sC_{gd})r_{2}+\frac{N_{2}(sC_{gd}-g_{ m_2})}{D_{1}}\Big] \\
T_{21}&= \frac {-g_{m_1}[sC_2+g_{m_2}+g_{2}]} {D_{1}}  \\
T_{22}&= \frac {N_{2}}{D_{1}} 
\end{align}
where 
\begin{align*}
D_{1}&=  [(C_1+C_2)C_{gd}+C_1C_2]s^2 +[C_2g_1+C_1g_2+(g_{m_2}-g_{m_1}+g_1+g_2)C_{gd}]s +[g_{m_1}g_{m_2}+g_1g_2] \\
N_{2}&=[1+(s(C_{1}+C_{gd})+g_{1})r_{2}][sC_2+g_{m_2}+g_2 ]- [sC_1+g_1-g_{ m_1}][-1+(g_{m_2}-sC_{gd})r_{2}] \\
N_{1}&=\frac{-sC_{x}}{(sC_{x}r_{1}+1)(sC_2+g_{m_2}+g_2)}
\end{align*}
The values of various parameters are as given below.
$g_{m_1}= 1.8\times10^{-3} A/V $\\ 
$g_{m_2}= 4\times10^{-5} A/V $\\ 
$g_1=\frac{1}{R_1}= \frac{1}{800\times10^3}= 1.25\times10^{-6} A/V $\\ 
$g_2=\frac{1}{R_2}= \frac{1}{300\times10^3}= 3.3333\times10^{-6} A/V $\\ 
$C_1= 0.5\times10^{-12} F $\\ 
$C_2= 68.48\times10^{-12} F $\\ 
$C_{gd}= 0.05\times10^{-12} F $ 

Using $r = r^{'}= 0.1 \ \Omega$, matrix $T$ can be regularized such that the $D$ matrix in its state space model exists which is non-singular and thus the matrix $T$ will have all proper elements.
he matrix $T$ can now be inverted. The elements of matrix $T$ are as given below.
\begin{align*}
T_{11}&=\frac {10s(s+2.327\times10^{6})(s+4.751\times10^{4})}{(s+2\times10^{14})(s^{2}-1.338\times10^{4}s +1.91\times10^{15})}\notag \\ 
T_{12}&=\frac {1.326\times10^{11}s(s+8.25\times10^{7})}{(s+2\times10^{14})(s^{2}-1.338\times10^{4}s +1.91\times10^{15})}\notag \\ 
\end{align*}
\begin{align*}
T_{21}&=\frac {-3.2705\times10^{9}(s+6.328\times10^{5})}{s^{2}-1.338\times10^{4}s +1.91\times10^{15}}\notag \\
T_{22}&=\frac {0.1(s+1.83\times10^{13})(s-2.545\times10^{7})}{s^{2}-1.338\times10^{4}s +1.91\times10^{15}}\notag
\end{align*}

The right coprime factorization of $T$ gives matrices $D_{r}$ and $N_{r}$ respectively. The elements of matrix $D_{r}$ are as given below.


\begin{align*}
D_{r 11}&= \frac {(s+2\times10^{ 14})(s+1.29\times10^{11})(s+3.87\times10^{8})}{(s+1\times 10^{ 10})(s+2\times 10^{10})(s+3\times 10^{12})}  \\
D_{r 12}&= \frac {-3.15\times10^{10}(s+1.17\times10^{14})(s+4.55\times 10^{8})}{(s+1\times 10^{10})(s+2\times 10^{10})(s+3\times 10^{12})} \\
D_{r 21}&=\frac {-1.68\times10^{13}(s+8.89\times10^{9})(s-9.43\times10^{7})}{(s+1\times 10^{10})(s+2\times 10^{10})(s+3\times 10^{12})} \\
D_{r 22}&=\frac {(s+9\times10^{9})(s-9.03\times10^{7})}{(s+1\times10^{10})(s+2\times10^{10})}
\end{align*}
The elements of matrix $N_{r}$ are as given below.
\begin{align*}
N_{r 11}&= \frac {10(s+1.28\times10^{11})(s+3.155\times10^{8})(s-0.3748)}{(s+1\times 10^{10})(s+2\times 10^{10})(s+3\times 10^{ 12})}  \\
N_{r 12}&=\frac {-1.83\times10^{11}(s-0.4025)(s+3.67\times10^{8})}{(s+1\times 10^{10})(s+3\times 10^{12})(s+2\times 10^{10})}  \\ 
N_{r 21}&=\frac {-1.68\times10^{12}(s+3.68\times10^{8})(s-0.4025)}{(s+1\times 10^{10})(s+2\times 10^{10})(s+3\times 10^{12})} \\ 
N_{r 22}&=\frac {0.1(s+1.83\times10^{13})(s+1.12\times10^{10})}{(s+1\times10^{10})(s+2\times10^{10})}
\end{align*}



By solving the Bezout's identity $X_{l}N_{r}+Y_{l}D_{r}=I$, we get $X_{l}$ and $Y_{l}$ respectively. The elements of matrix $X_{l}$ are as given below.
\begin{align*}
X_{l 11}&=\frac {1.41\times10^{13}(s+4.36\times10^{14})(s+2.19\times10^{11})}{(s+1\times 10^{11})(s+2\times 10^{12})(s+3\times 10^{13})} \\
X_{l 12}&= \frac {-5.48\times10^{15}(s+2\times10^{7})}{(s+1\times 10^{11})(s+3\times 10^{13}) } \\
X_{l 21}&= \frac {2.02\times10^{14}(s+1.98\times10^{14})(s+2.09\times10^{11})}{(s+1\times 10^{11})(s+2\times 10^{12})(s+3\times 10^{13})}\\
X_{l 22}&=\frac {-3.28\times10^{16}(s+1.92\times10^{7})}{(s+1\times 10^{11})(s+3\times 10^{13})} 
\end{align*}
The elements of matrix $Y_{l}$ are as given below.
\begin{align*}
Y_{l 11}&= \frac {(s-3.06\times10^{14})(s^{2}+2.76\times10^{11}s+1.22\times10^{23})}{(s+1\times 10^{11})(s+2\times 10^{12})(s+3\times 10^{13})}\\
Y_{l 12}&= \frac {5.48\times10^{14})(s+1.83\times10^{13})}{(s+1\times 10^{11})(s+3\times 10^{13}}\\
Y_{l 21}&=\frac {-2\times10^{15}(s^{2}+2.6\times10^{11}s+1.11\times10^{23})}{(s+1\times 10^{11})(s+2\times 10^{12})(s+3\times 10^{13} } \\
Y_{l 22}&=\frac {(s+3.28\times10^{15})(s+1.82\times10^{13})}{(s+1\times 10^{11})(s+3\times 10^{13})}
\end{align*}


For $Q=0$ the stabilizing compensator $T_{c}$ is given as $X_{l}^{-1}Y_{l}$. Using $X_{l}$ and $Y_{l}$ as computed above we get the elements of $T_{c}$ as given below.
\begin{align*}
T_{c 11}&=\frac {-5.05\times10^{-14}(s^{2}+1.05\times10^{9}s+1.15\times10^{19})}{s+4.104\times10^{7}}\notag \\ 
T_{c 12}&=\frac {8.46\times10^{-15}(s+2\times10^{12})(s+1.92\times10^{10})}{s+4.104\times10^{7}}\notag \\
T_{c 21}&=\frac {-3.12\times10^{-16}(s+3.51\times10^{12})(s-2.97\times10^{12})}{s+4.104\times10^{7}}\notag \\
T_{c 22}&=\frac {2.17\times10^{-17}(s-4.19\times10^{15})(s+2\times 10^{13})}{s+4.104\times10^{7}}\notag
\end{align*}

Now let us find $\hat T$ which is $(T^{-1}+T_{c}^{-1})^{-1}$. The elements of $\hat T$ are as given below. 

\begin{equation*}
\begin{bmatrix}
\hat T_{11} & \hat T_{11} \\
\hat T_{11} & \hat T_{11}
\end{bmatrix} =
\begin{bmatrix}
\frac{\hat N_{11}}{\hat D_{11}} & \frac{\hat N_{11}}{\hat D_{11}} \\ \\
\frac{\hat N_{11}}{\hat D_{11}} & \frac{\hat N_{11}}{\hat D_{11}}
\end{bmatrix}
\end{equation*}
where
\begin{align*}
\hat N_{11}&=10(s-3.06\times10^{14})(s+5.83\times10^{5})(s^{2}+2.65\times 10^{11}s+1.21\times 10^{23})(s+6295) \\
\hat D_{11}&=\hat D_{22}=(s+1\times 10^{10})(s+1\times 10^{ 11})(s+2\times 10^{12})(s+3\times 10^{12})(s+3\times 10^{13}) \\   
\hat N_{12}&=5.49\times 10^{15}(s+1.83\times10^{13})(s+6.59\times10^{ 5})(s+52.37) \\
\hat D_{12}&= (s+1\times 10^{10})(s+1\times 10^{11})(s+3\times 10^{ 12})(s+3\times 10^{13}) \\
\hat N_{21}&=(s+2.2\times10^{7})(s^{2}+5.05\times 10^{11}s+1.3\times 10^{23}) \\
\hat D_{21}&=(s+1\times 10^{10})(s+1\times 10^{11})(s+2\times 10^{10})(s+2\times 10^{12})(s+3\times 10^{12})(s+3\times 10^{13}) \\
\hat N_{22}&= 0.1(s+3.29\times10^{15})(s+1.83\times10^{13})(s+1.83\times 10^{13})(s+3.73\times 10^{10})(s-4.23\times10^{6})  
\end{align*}

It can be seen that the elements of $\hat T$ belong to $M(\stabS)$ and the compensating network $T_{c}$ stabilizes the given network $T$.

\section{Conclusion}
We have developed a theory for compensation of a linear active network at its ports by another linear active network such that the interconnection is stable in the BSBR sense at these ports even when the parameters of the original network are not exact but can be anywhere in a sufficiently small neighbourhood. Our theory can be seen as an extension of the algebraic theory of feedback stabilization which has been well known \cite{vids} in control theory. While in the feedback stabilization theory the given linear system and the controller form a feedback loop, in the case of networks connected at ports, such a loop is not readily available. However the stable coprime fractional approach originally developed for feedback stabilization carries over to solve the problem. Theory of active network synthesis cannot be developed without the stabilization theory and a lack of suitable approach for synthesis of port compensation with stability has been possibly the main hurdle. The resulting stable interconnection is described by an affine parametrization in which the free parameter is itself a stable network function. This parametrization is analogous to the well known parametrization in feedback systems theory and hence has opened doors to approach active network synthesis using analytical methods such as H-infinity optimization. 







\end{document}